\documentclass[journal,twocolumn,12pt]{IEEEtran}
\usepackage[english]{babel}
\usepackage[latin1]{inputenc}
\usepackage{amssymb,amstext,amsmath}
\usepackage[ruled,vlined]{algorithm2e}
\usepackage{verbatim}
\usepackage{float}

\usepackage{cite}
\ifCLASSINFOpdf
   \usepackage[pdftex]{graphicx}
   \DeclareGraphicsExtensions{.pdf,.jpeg,.png}
\else
   \usepackage[dvips]{graphicx}
   \DeclareGraphicsExtensions{.eps}
\fi
\newtheorem{theorem}{Theorem}

\newtheorem{corollary}{Corollary}

\newtheorem{proof}{Proof}

\begin{document}

\title{Multiuser I-MMSE}
\author{Samah A. M. Ghanem, \textit{Senior Member, IEEE}\\
}
\maketitle
\begin{abstract}
\boldmath
In this paper, we generalize the fundamental relation between the derivative of the mutual information and the minimum mean squared error (MMSE) to multiuser setups. We prove that the derivative of the mutual information with respect to the signal to noise ratio (SNR) is equal to the MMSE plus a covariance induced due to the interference, quantified by a term with respect to the cross correlation of the multiuser input estimates, the channels and the precoding matrices. We also derive new relations for the gradient of the conditional and non-conditional mutual information with respect to the MMSE. Capitalizing on the new fundamental relations, we derive closed form expressions of the mutual information for the multiuser channels, particularly the two user multiple access Gaussian channel driven by binary phase shift keying (BPSK) to illustrate and shed light on methods to derive similar expressions for higher level constellations. We capitalize on the new unveiled relation to derive the multiuser MMSE and mutual information in the low-SNR regime.
\end{abstract}
\begin{IEEEkeywords}
Estimation Theory; Gradient of conditional mutual information; Gradient of non-conditional mutual information; Gradient of joint mutual information; Information Theory; Interference; MAC; MMSE; Mutiuser I-MMSE; Mutual Information.
\end{IEEEkeywords}


\section{Introduction}
Connections between information theory and estimation theory dates back to the work of Duncan, in \cite{Duncan} who showed that for the continuous-time additive white Gaussian noise (AWGN) channel, the filtering minimum mean squared error 
(causal estimation) is twice the input output mutual information for any underlying signal distribution. Recently, Guo, Shamai, and Verdu have illuminated intimate connections between information theory and estimation theory in a seminal paper, \cite{34}. In particular, Guo et al. have shown that in the classical problem of information transmission through the conventional AWGN channel, the  derivative of the mutual information with respect to the SNR is equal to the smoothing minimum mean squared error (noncausal estimation); a relationship that holds for scalar, point-to-point vector, discrete-time and continuous-time channels regardless of the input statistics. The relevance of these recent connections comes from the fact that mutual information and MMSE are two canonical operational measures in information theory and estimation theory: mutual information measures the reliable information transmission rate between the input and the output of a system for a specific signaling scheme, while MMSE measures the minimum mean squared error in estimating the input given the output. Later Palomar and Verdu generalized this relation to linear vector Gaussian channels \cite{35}, \cite{36}.  The mutual information was also represented as an integral of a certain measure of the estimation error in Poisson channels \cite{67}, \cite{68}. There have been extensions of these results to the case of mismatched input distributions in the scalar Gaussian channel in \cite{VerduS} and \cite{weissman}. Most recently, Ghanem in \cite{98}, \cite{newSamah}, derived the gradient of the mutual information with respect to arbitrary parameters for the multiple access Gaussian channels, a relation that extends the relation for the case of mutually interfering inputs in linear vector Gaussian channels to the case of multiple non-mutually interfering inputs and with mutual interference, a starting point to the results in this work. The implications of a framework involving key quantities in information theory and estimation theory are countless both from the theoretical \cite{37}, \cite{38} and the more practical perspective, \cite{43}, \cite{39}, \cite{107}, \cite{SamahWiMob16}, \cite{SamahMCPjournal}.

The intimate connection between information measures and estimation measures allow few explicit closed form expressions of the mutual information for binary inputs to be derived, particularliy ones for BPSK and QPSK over the Single Input Signle Output (SISO) channel,~\cite{34}, \cite{samahMAPtele}, \cite{samahthesis}. Therefore, it is of particular importance to address connections between information theory and estimation theory for the multiuser case in order to understand the communication framework under such inputs and try to provide explicit forms when multiple accessing and interfering inputs coexist.

In this paper, we first revisit the connections between the mutual information and the MMSE for the multiuser setup, see also \cite{newSamah},~\cite{SamahMCPjournal}. Therefore, the fundamental relation between the derivative of the mutual information and the MMSE, known as I-MMSE identity, and defined for point to point channels with any noise or input distributions in \cite{34} is not anymore suitable for the multiuser case. Therefore, we generalize the I-MMSE relation to the multiuser case. Moreover, we generalize the relations for linear vector Gaussian channels in \cite{35} to multiuser channels where we extend these relations to the per-user gradient of the mutual information with respect to the MMSE, channels and precoders (power allocation) matrices of the user and the interferers. Then, we derive new closed form expressions for the mutual information for single user and mutiuser scalar Gaussian channels driven by BPSK inputs. Further, we analyze the MAC Gaussian channel model at the asymptotic regime of low SNR and capitalize on the new unveiled connections between the mutual information and the MMSE to derive the low SNR expansion of the mutual information in a multiuser setup. 

The implications of the derived relations in its two-user's version presented in this paper, or in its $K$-users version \cite{Ghanem16}, where the effect of $K-1$ interferers is also characterized, are many-fold, whether to characterize the capacity of interference channels, or to characterize novel schemes that extends state of art power allocation \cite{39} to ones that can capitalize on the characterization of the interference effect \cite{98}. Additionally, the extension of such relation to wireless networks with noisy coded flows is of particular importance \cite{GhanemNetCode16}, \cite{GhanemMedardJournal2017}. For instance, a novel piggybacking capacity achieving scheme is recently proposed for networks with Amplify and Forward (AF), \cite{GhanemMedardISIT17} capitalizing on this result and its characterization of the gap from the cut-set upper bound.

Throughout the paper, the following notation is employed, boldface uppercase letters denote matrices, lowercase letters denote scalars. The superscript, $(.)^{-1}$, $(.)^{T}$, $(.)^{*}$, and $(.)^{\dag}$ denote the inverse, transpose, conjugate, and conjugate transpose operations. The $(\nabla)$ denotes the gradient of a scalar function with respect to a variable.  The $\mathbb{E}[.]$ denotes the expectation operator. The $||.||$ and $Tr\left\{.\right\}$ denote the Euclidean norm, and the trace of a matrix, respectively.

The rest of the paper is organized as follows; section II introduces the system model. Section III introduces the new fundamental relations between the mutual information and the MMSE. Section IV introduces the new closed form expression of the mutual inforamtion. Section V introduces analysis at the asymptotic regime of low SNR. Section VI concludes the paper.

\section{System Model}
Consider the deterministic complex-valued vector channel,
\begin{equation}
\label{1}
{\bf{y}}=\sqrt{snr}~{\bf{H_1P_1}}{\bf{x_1}}+\sqrt{snr}~{\bf{H_2P_2}}\bf{x_2}+\bf{n},
\end{equation}

where the $n_r~\times~1$ dimensional vector $\bf{y}$ and the $n_t~\times~1$ dimensional vectors $\bf{x_1}$, $\bf{x_2}$ represent, respectively, the received vector and the independent zero-mean unit-variance transmitted information vectors from each user input to the MAC channel. The distributions of both inputs are not fixed, not necessarily Gaussian nor identical. The $n_r \times n_t$ complex-valued  matrices $\bf{H_1}$, $\bf{H_2}$ correspond to the deterministic channel gains for both input channels (known to both encoder and decoder) and $\bf{n} \sim \mathcal{CN}(0,I)$ is the $n_r \times 1$ dimensional complex Gaussian noise with independent zero-mean unit-variance components. 

\section{New Fundamental Relations between the Mutual Information and the MMSE}
The first contribution is given in the following theorem, which provides a generlization of the I-MMSE identity to the multiuser case and traverses back to the same identilty of the single user case.
\begin{theorem}
\label{theorem1.0}
The relation between the derivative of the joint mutual information with respect to the snr and the total non-linear MMSE for a multiuser Gaussian channel satisfies:
\begin{equation}
\frac{dI(snr)}{dsnr}=mmse(snr)+\psi(snr)
\end{equation}
\end{theorem}
Where,
\vspace{-0.15cm}
\small
\begin{multline}
\label{mmsee}
mmse(snr)=Tr\left\{\bf{H_1P_1}\bf{E_1}\bf{(H_1P_1)^{\dag}}\right\}+\\
Tr\left\{\bf{H_2P_2}\bf{E_2}\bf{(H_2P_2)^{\dag}} \right\},
\end{multline}
\normalsize
\scriptsize
\begin{multline}
\label{psii}
\psi(snr)=-Tr\left\{\bf{H_1P_1}\mathbb{E}_y[\bf{\mathbb{E}_{x_1|y}[\bf{x_1|y}]\mathbb{E}_{x_2|y}[\bf{x_2|y}]^{\dag}}]\bf{(H_2P_2)^{\dag}}\right\} \nonumber \\
-Tr\left\{\bf{H_2P_2}\mathbb{E}_y[\bf{\mathbb{E}_{x_2|y}[\bf{x_2|y}]\mathbb{E}_{x_1|y}[\bf{x_1|y}]^{\dag}}]\bf{(H_1P_1)^{\dag}}\right\},
\end{multline}
\normalsize
\begin{proof}
See Appendix A.
\end{proof}
The per-user covariance matrix of the estimation error, also called the per-user MMSE matrix is given respectively as follows:
\begin{equation}
\label{7}
\bf{E_1}=\bf{\mathbb{E}_y[(x_1-\widehat{x}_1)(x_1-\widehat{x}_1)^{\dag}]}
\end{equation}
\begin{equation}
\label{8}
\bf{E_2}=\bf{\mathbb{E}_y[(x_2-\widehat{x}_2)(x_2-\widehat{x}_2)^{\dag}]}.
\end{equation}
The input estimates of each user input is given respectively as follows:
\small
\begin{multline}
\label{eq.37}
\widehat{{\bf{x}}}_1={\bf{\mathbb{E}_{x_{1}|y}}[{\bf{x_{1}|y}}]}\\
=\sum_{{\bf{x_{1},x_{2}}}} \frac{{\bf{x_{1}}}p_{y|x_{1},x_{2}}({\bf{y|x_{1},x_{2}}})p_{x_1}({\bf{x_{1}}})p_{x_2}({\bf{x_{2}}})}{p_{y}({\bf{y}})}
\end{multline}
\begin{multline}
\label{eq.38}
\widehat{{\bf{x}}}_2={\bf{\mathbb{E}_{x_{2}|y}}[{\bf{x_{2}|y}}]}\\
=\sum_{{\bf{x_{1},x_{2}}}} \frac{{\bf{x_{2}}}p_{y|x_{1},x_{2}}({\bf{y|x_{1},x_{2}}})p_{x_1}({\bf{x_{1}}})p_{x_2}({\bf{x_{2}}})}{p_{y}({\bf{y}})}.
\end{multline}
\normalsize
The conditional probability distribution of the Gaussian noise is defined as:
\small
\begin{equation}
p_{y|x_1,x_2}({\bf{y|x_1, x_2}})=\frac{1}{\pi^{n_r}}e^{-\left\|{\bf{y}}-\sqrt{snr}{\bf{H_1P_1x_1}}-\sqrt{snr}{\bf{H_2P_2x_2}}\right\|^{2}}
\end{equation}
\normalsize
The probability density function for the received vector $\bf{y}$ is defined as:
\small
\begin{equation}
{p_{y}({\bf{y}})}=\sum_{{\bf{x_{1},x_{2}}}} {p_{y|x_{1},x_{2}}}({\bf{y|x_{1},x_{2}}})p_{x_1}({\bf{x_{1}}})p_{x_2}({\bf{x_{2}}}).
\end{equation}
\normalsize
Henceforth, for the case of two-user MAC, the system MMSE with respect to the SNR $mmse(snr)$ is the MMSE corresponding to the best estimation of  inputs $\bf{x_1}$ and $\bf{x_2}$ upon the observation for a given signal-to-noise ratio (SNR), i.e.,
\small
\begin{multline}
\label{14}
mmse(snr)=\bf{\mathbb{E}_y}\left[\left\|\bf{H_1P_1}(\bf{x_1}-\bf{\mathbb{E}_{x_1|y}[x_1|y]})\right\|^{2}\right]\\
+\bf{\mathbb{E}_y}\left[\left\|\bf{H_2P_2}(\bf{x_2-\mathbb{E}_{x_2|y}[x_2|y]})\right\|^{2}\right],
\end{multline}
\normalsize
\small
\begin{equation}
~~=Tr\left\{\bf{H_1P_1}\bf{E_1}{(\bf{H_1P_1})}^{\dag}\right\} + Tr\left\{{\bf{H_2P_2}\bf{E_2}{(\bf{H_2P_2})}^{\dag}}\right\}  
\end{equation}
\normalsize
as given in Theorem~\ref{theorem1.0}.

Note that $I(snr)$ is the joint mutual information in $I(\bf{x_1,x_2,y})$, the term $mmse(snr)$ is due to the users MMSEs, particularly, $mmse(snr)=mmse_1(snr)+mmse_2(snr)$ and $\psi(snr)$ are covariance terms that appear due to the covariance of the interferers. Those terms are with respect to the channels, precoders, and non-linear estimates of the user inputs. When the covariance terms vanish to zero, the mutual information will be equal to the mmse, with respect to the SNR. This applies to the single user and point to point communications. Therefore, the result of Theorem~\ref{theorem1.0} is a generalization of previous result and boils down to the result of Guo et al, \cite{34} under certain conditions which are: (i) when the cross correlation between the inputs estimates equals zero (ii) when interference can be neglected or easily removed (i.e. interference is very weak, very strong, or aligned) (iii) under the signle user setup. (iv) when certain access and power allocation scheme is used for inputs Gaussian distributed and a successive clean interference estimation process is performed, \cite{GhanemMedardISIT17}.

Such generalized fundamental relation between the change in the mutliuser mutual information and the SNR is of particular relevance. Firstly, such result allows us to understand the behavior of per-user rates with respect to the interference due to the mutual interference and the interference due to other users behaviour in terms of power levels and channel strengths. In addition, the result allows us to be able to quantify the losses incured due to the interference in terms of bits. Therefore, when the term $\psi(snr)$ equals zero. The derivative of the mutual information with respect to the SNR equals the total $mmse(snr)$:
\begin{equation}
\frac{dI(snr)}{d snr}=mmse(snr),
\end{equation}
which matches the result by Guo et. al in~\cite{34}.

\subsection{The Conditional and Non-Conditional I-MMSE and a Remark on Interference Channels}
The implication of the derived relation on the interference channel is of particular relevance. To particular, it is worth to note that we can capitalize on the new fundamental relation to extend the derivative with respect to the SNR to the conditional and non-conditional mutual information components that provides per user rates. To make this more clear, we capitalize on the chain rule of the mutual information which states the following:
\begin{equation}
\label{chainR}
I({\bf{x_1, x_2; y}})=I({\bf{x_1; y}})+I({\bf{x_2; y|x_1}}) 
\end{equation}
Therefore, through this observation we can conclude the following theorem.
\begin{theorem}
\label{theoremConditional}
The relation between the derivative of the conditional and the non-conditional mutual information and their corresponding minimum mean squared error satisfies, respectively:
\begin{equation}
\frac{dI(\bf{x_2; y}|x_1)}{dsnr}=mmse_2(snr)+\psi(snr)
\end{equation}
\begin{equation}
\frac{dI({\bf{x_1; y}})}{dsnr}=mmse_1(\gamma snr)
\end{equation}
\end{theorem}
\begin{proof}
Taking the derivative of both sides of \eqref{chainR}, and subtracting the derivative of $I(\bf{x_1; y})$ which is equal to $mmse_1(\gamma snr)$, $\gamma$ is a scaling factor, due to the fact that $x_1$ is decoded first considering the other users' input $x_2$ as noise. Therefore, Theorem~\ref{theoremConditional} has been proved.
\end{proof}
Note that the derivative of the conditional mutual information as well as the non-conditional mutual information can be scaled due to different SNRs in a two-user interference channel. However, the scaling is straightforward to apply. For further details, refer to \cite{samahthesis}, \cite{SamahMCPjournal}.

The following theorems in addition to Theorem \ref{theorem1.0} generalizes the connections between information theory and estimation theory to the multiuser case.
\begin{theorem}
\label{theorem1.4}
The relation between the gradient of the mutual information with respect to the channel and the non-linear MMSE for a two user-MAC channel with arbitrary inputs~\eqref{1} satisfies:
\small
\begin{equation}
\label{eq.5}
{{\bf{\nabla_{H_1}}} I({\bf{x_{1},x_{2};y}})=\bf{H_1P_1E_1P_1}^{\dag}-\bf{H_2 P_2}\bf{\mathbb{E}[\widehat{x}_2\widehat{x}_1^{\dag}]}\bf{P_1}^{\dag}}
\end{equation}
\begin{equation}
\label{eq.6}
{{\bf{\nabla_{H_2}}} I({\bf{x_{1},x_{2};y}})=\bf{H_2P_2E_2P_2^{\dag}}-\bf{H_1 P_1} \bf{\mathbb{E}[\widehat{x}_1\widehat{x}_2^{\dag}]}\bf{P_2}^{\dag}}
\end{equation}
\normalsize
\end{theorem}
\begin{proof}
The detailed proof has been provided in Appendix B, \cite{newSamah}
\end{proof}
\begin{theorem}
\label{theorem2.4}
The relation between the gradient of the mutual information with respect to the precoding matrix and the non-linear MMSE for a two user-MAC channel with arbitrary inputs~\eqref{1} satisfies:
\small
\begin{equation}
\label{eq.5.1}
{\bf{\nabla_{P_1}}}I({\bf{{x_{1},x_{2};y}}})=\bf{H_1^{\dag}}{\bf{H}_{1}}{\bf{P}_1}{\bf{E}_{1}}-\bf{H_1^{\dag}}{\bf{H}_{2}}{\bf{P}_2}\bf{\mathbb{E}[\widehat{x}_2\widehat{x}_1^{\dag}]}
\end{equation}
\begin{equation}
\label{eq.6.1}
{\bf{\nabla_{P_2}}}I({\bf{x_{1},x_{2};y}})=\bf{H_2^{\dag}}{\bf{H}_{2}}{\bf{P}_2}{\bf{E}_{2}}-\bf{H_2^{\dag}}{\bf{H}_{1}}{\bf{P}_1}\bf{\mathbb{E}[\widehat{x}_1\widehat{x}_2^{\dag}]}
\end{equation}
\normalsize
\end{theorem}
\begin{proof}
The detailed proof has been provided in Appendix C, \cite{newSamah}
\end{proof}
Theorems~\ref{theorem1.4} and \ref{theorem2.4} provide intuitions about the change of the mutual information with respect to the changes in the channel or the precoding (power allocation). A straightforward connection between both changes when each gradient is scaled with respect to the changing arbitrary parameter, we can write this connection as follows:
\begin{equation}
 {\bf{\nabla_{P_1}}}I({\bf{{x_{1},x_{2};y}}}){\bf{P_1^{\dag}}}={\bf{H_1^{\dag}}}{\bf{\nabla_{H_1}}}I({\bf{x_{1},x_{2};y}})
\end{equation}
\begin{equation}
 {\bf{\nabla_{P_2}}}I({\bf{x_{1},x_{2};y}}){\bf{P_2^{\dag}}}={\bf{H_2^{\dag}}}{\bf{\nabla_{H_2}}}I({\bf{x_{1},x_{2};y}})
\end{equation}
Note that we can derive the gradient of the mutual information with respect to any arbitrary parameter following similar steps of the proof of the previous two theorems. Note also that the derived relations in \eqref{eq.5},~\eqref{eq.6},~\eqref{eq.5.1}, and \eqref{eq.6.1} reduce to the relation between the gradient of the mutual information and the non-linear MMSE derived for the linear vector Gaussian channels \cite{35} if the cross correlation between the input estimates is zero, which applies to linear estimation with perfect reconstruction of estimates and removal of one estimate from another, \cite{GhanemMedardISIT17}.
\subsection{Gradient of the Conditional and Non-Conditional Mutual Information}
Theorem~\ref{theorem2.4} shows how much rate is lost due to the other user. The gradient of the mutual information provides a set of terms that are associated to the mutual interference, which provides a positive change, however, the loss is attributed to the effect of the non-mutual interference, which appears in the second term as a negative change. Therefore, we can account for such quantified rate loss via optimal power allocation and optimal precoding. In order to be able to understand more deeply the achieved rates of each user in a MAC channel. Ee capitalize on the chain rule of the mutual information to derive the conditional mutual information as follows,
\begin{equation}
\label{a390}
I({\bf{x_{1},x_{2};y}})=I({\bf{x_{2};y}})+I({\bf{x_{1};y|x_{2}}})
\end{equation}
and,
\begin{equation}
\label{a391}
I({\bf{x_{1},x_{2};y}})=I({\bf{x_{1};y}})+I({\bf{x_{2};y|x_{1}}})
\end{equation}
Where the joint mutual information term is defined as follows:
\vspace{-0.15cm}
\small
\begin{equation}
I({\bf{x_1,x_2;y}})= \mathbb{E}\left[\frac{p_{y|x_1,x_2}({\bf{y|x_1,x_2}})}{\sum_{{\bf{x_1',x_2'}}}p_{{\bf{y|x_1',x_2'}}}({\bf{y|x_1',x_2'}})p_{{\bf{x_1'}}}({\bf{x_1'}})p_{{\bf{x_2'}}}({\bf{x_2'}})}\right]
\end{equation}
\normalsize
Where, $x_1'$ and $x_1'$ correspond to all possible permutations of $x_1$ and $x_2$ drawn from each inputs' constellation set. 
The non-conditional mutual information is defined as follows:
\begin{equation}
I({\bf{x_1;y}})= \mathbb{E}\left[\frac{p_{y|x_1}({\bf{y|x_1}})}{\sum_{{\bf{x_1'}}}p_{{\bf{y|x_1'}}}({\bf{y|x_1'}})p_{{\bf{x_1'}}}({\bf{x_1'}})}\right]
\end{equation}
Where the signal $\bf{x}_2$ is considered as noise.
\begin{equation}
I({\bf{x_2;y}})= \mathbb{E}\left[\frac{p_{y|x_2}({\bf{y|x_2}})}{\sum_{{\bf{x_2'}}}p_{{\bf{y|x_2'}}}({\bf{y|x_2'}})p_{{\bf{x_2'}}}({\bf{x_2'}})}\right]
\end{equation}
Where the signal $\bf{x}_1$ is considered as noise. The conditional mutual information is defined as follows:
\vspace{-0.15cm}
\small
\begin{equation}
I({\bf{x_1;y|x_2}})= \mathbb{E}\left[\frac{p_{y|x_1',x_2}({\bf{y|x_1',x_2}})}{\sum_{{\bf{x_1',x_2}}}p_{{\bf{y|x_1',x_2}}}({\bf{y|x_1',x_2}})p_{{\bf{x_1'}}}({\bf{x_1'}})p_{{\bf{x_2}}}({\bf{x_2}})}\right]
\end{equation}
\begin{equation}
I({\bf{x_2;y|x_1}})= \mathbb{E}\left[\frac{p_{y|x_1,x_2'}({\bf{y|x_1,x_2'}})}{\sum_{{\bf{x_1,x_2'}}}p_{{\bf{y|x_1,x_2'}}}({\bf{y|x_1,x_2'}})p_{{\bf{x_1}}}({\bf{x_1}})p_{{\bf{x_2'}}}({\bf{x_2'}})}\right]
\end{equation}
\normalsize
Clearly, we know that $I({\bf{x_{2};y}})$ is the mutual information when user 2 is decoded first considering user 1 signal as noise. Therefore, we can write it as follows,
\vspace{-0.1cm}
\begin{equation}
\label{a41}
I({\bf{x_{2};y}})=\mathbb{E}\left[log \frac{p_{y|x_{2}}({\bf{y|x_{2}}})}{p_{y}({\bf{y}})}\right]
\end{equation}
\scriptsize
\begin{multline}
\label{a42}
p_{y|x_{2}}({\bf{y|x_{2}}})=\frac{1}{\pi^{n_r}} \times \\
e^{-{({\bf{y}}-\sqrt{snr}\bf{{H}_{2}}{{P_2}}\bf{x_{2}})^{\dag}(\bf{{P}_{1}^{\dag}{H}_{1}^{\dag}{H}_{1}{P_1}+I})^{-1}({\bf{y}}-\sqrt{snr}\bf{{H}_{2}}{{P_2}}\bf{x_{2}})}}
\end{multline}
\normalsize
\begin{equation}
\label{a43}
p_{y}({\bf{y}})=\displaystyle\sum\limits_{x_{2}'} p_{y|x_{2}'}({\bf{y|x_{2}'}})p_{x_{2}'}({\bf{x_{2}'}})
\end{equation}
Based on such definition, we conclude the following theorem which provides a new fundamental relation between the gradient of the mutual information with respect to the precoder of the other user given that the other user will be secondly decoded.
\begin{theorem}
\label{gradientwithnoise}
The gradient of the mutual information with respect to the precoder, for a scaled user power when the other user input is considered as noise is as follows,
\scriptsize
\begin{equation}
{{\nabla_{{\bf{P_1}}}}}I({\bf{x_{2};y}})=\bf{{H}_{2} P_2 E_{2}}\bf{{P}_{2}^{\dag}{H}_{2}^{\dag}} \bf{{H}_{1}}^{\dag}\bf{H_{1}P_{1}}({\bf{{P}_{1}^{\dag}{H}_{1}^{\dag}{H}_{1}{P_1}+I}})^{-1}
\end{equation}
\begin{equation}
 {\nabla_{{\bf{P_2}}}}I({\bf{x_{1};y}})=\bf{{H}_{1} P_1 E_{1}} \bf{{P}_{1}^{\dag}{H}_{1}^{\dag}} \bf{{H}_{2}}^{\dag}\bf{H_{2}P_{2}}({\bf{{P}_{2}^{\dag}{H}_{2}^{\dag}{H}_{2}{P_2}+I}})^{-1}
\end{equation}
\normalsize
\end{theorem}
\begin{proof}
The proof follows similar steps of the proof of Theorem~\ref{theorem2.4}. 
\end{proof}
When each user is transmitting over a single channel, the new relation in Theorem \ref{gradientwithnoise} will be more clearly understood in terms of the effect of the interference plus noise power scaling on the gradient, see \cite{SamahMCPjournal}. In other words, Theorem \ref{gradientwithnoise} is of particular relevance to understand how the rate changes and can be adapted based on the changes of the channel and power of the interference.
\begin{corollary}
The gradient of the conditional mutual information will be as follows,
\small
\begin{equation}
\label{a260}
 {\bf{\nabla_{P_1}}}I({\bf{x_{1};y}|x_{2}})= {\bf{\nabla_{P_1}}}I({\bf{x_{1},x_{2};y}})- {\bf{\nabla_{P_1}}}I({\bf{x_{2};y}})
\end{equation}
\begin{equation}
\label{a260b1}
 {\bf{\nabla_{P_2}}}I({\bf{x_{1};y}|x_{2}})= {\bf{\nabla_{P_2}}}I({\bf{x_{1},x_{2};y}})- {\bf{\nabla_{P_2}}}I({\bf{x_{2};y}})
\end{equation}
\begin{equation}
\label{a271b1}
 {\bf{\nabla_{P_1}}}I({\bf{x_{2};y}|x_{1}})={\bf{\nabla_{P_1}}}I({\bf{x_{1},x_{2};y}})-{\bf{\nabla_{P_1}}}I({\bf{x_{1};y}})
\end{equation}
\begin{equation}
\label{a271}
 {\bf{\nabla_{P_2}}}I({\bf{x_{2};y}|x_{1}})={\bf{\nabla_{P_2}}}I({\bf{x_{1},x_{2};y}})-{\bf{\nabla_{P_2}}}I({\bf{x_{1};y}})
\end{equation}
\normalsize
\end{corollary}
\begin{proof}
The proof of the corollary follows from the chain rule of mutual information and the gradient of the mutual information derived for the sum rate and non-conditional rates. The second term in \eqref{a260} and \eqref{a271} are given in Theorem~\ref{gradientwithnoise}. The second term of \eqref{a260b1} and \eqref{a271b1} are provided with proof in \cite{SamahMCPjournal}.
\end{proof}
\section{Mutiuser MMSE and Mutual information Closed Forms}
The only known explicit closed forms of the MMSE and the mutual information are for BPSK inputs, \cite{34}, and QPSK \cite{samahMAPtele} for the SISO channel. In the SISO case, the relation between the mutual information and the MMSE for the signle user setup allows for the derivation of this form. However, for a MAC channel, as an example of multiuser channels, we will first consider a unit power, unit channel gains for simplicity. We will then capitalize on the new unveiled multiuser I-MMSE generalization of the relation between the mutual information and the two user MMSE with the covariance. Therefore, we derive new explicit closed form expressions of the MMSE and the mutual information for each user in the two user Gaussian MAC driven by BPSK. To derive a closed form expression of the conditional and non-conditional mutual information for each user under the MAC, we capitalize again on the chain rule of the mutual information stated in \eqref{a391}. The first user which will be decoded first given that the other user is noise, therefore the MMSE and the mutual information of user 1 will be respectively, given by the following theorems.
\begin{theorem}
\label{thm6}
The non-conditional $mmse_1'(snr)$ of the user 1 decoded first and scaled with the other user as noise is given by:
\vspace{-0.15cm}
\small
\begin{equation}
mmse_1'(snr)=1-\frac{1}{4\sqrt{\pi}}\int_{y \in \textsl{R}}tanh\left(\frac{\sqrt{snr}}{2}\bf{y}\right)e^{\frac{-\left({\bf{y}}-\sqrt{snr}\right)^{2}}{4}}d\bf{y}
\end{equation}
\end{theorem}
\normalsize
\begin{proof}
See Appendix B, part I.
\end{proof}
\begin{theorem}
\label{thm7}
The non-conditional mutual information $I_1'(snr)$ of the user 1 decoded first and scaled with the other user noise is given by:
\vspace{-0.05cm}
\small
\begin{equation}
I_1'(snr)=\frac{snr}{4}-\frac{1}{4\sqrt{\pi}}\int_{y \in \textsl{R}}log~cosh\left(\frac{\sqrt{snr}}{2}{\bf{y}}\right)e^{\frac{-\left({\bf{y}}-\sqrt{snr}\right)^{2}}{4}} d\bf{y}
\end{equation}
\end{theorem}
\normalsize
\begin{proof}
See Appendix B, part II.
\end{proof}
However, the conditional MMSE and conditional mutual information of user 2 that will be decoded next under the MAC given that the first user is decoded first are given on the following theorems.
\begin{theorem}
\label{thm8}
The conditional $mmse_2(snr)$ of user 2 decoded second given that user 1 in the MAC channel is decoded first with BPSK inputs is given by:
\vspace{-0.15cm}
\small
\begin{equation}
mmse_2(snr)=1-\frac{1}{\sqrt{2\pi}}\int_{y \in \textsl{R}}{tanh\left(\sqrt{{snr}}\bf{y}\right)e^{\frac{-\left({\bf{y}}-\sqrt{snr}\right)^{2}}{2}}d\bf{y}}
\end{equation}
\end{theorem}
\normalsize
\begin{proof}
The proof follows similar steps as in Appendix B, Part I, and follows the formula in \cite{34}.
\end{proof}
\begin{theorem}
\label{thm9}
The conditional mutual information $I_2'(snr)$ of user 2 decoded second given that user 1 in the MAC channel is decoded first with BPSK inputs is given by:
\vspace{-0.05cm}
\small
\begin{equation}
I_2'(snr)={snr}-\frac{1}{\sqrt{2\pi}}\int_{y \in \textsl{R}}log~cosh\left(\sqrt{{snr}}{\bf{y}}\right)e^{\frac{-\left({\bf{y}}-\sqrt{snr}\right)^{2}}{2}} d\bf{y}
\end{equation}
\end{theorem}
\normalsize
\begin{proof}
The proof follows similar steps as in Appendix B, Part II, and follows the formula in \cite{34}.
\end{proof}
Notice that if both users are time sharing or decoded jointly, at such point, a maximum sum rate is acheivable, therefore, each user's rate will follow the one in Theorem~\ref{thm9}. Such case is similar to two parallel channels for each user, therefore, the sum rate is the sum of each individual rate. However, from Theorems~\ref{thm7} to \ref{thm8}, its straigntforward to conclude the following corollaries that defines the total MMSE and mutual information of a two user MAC driven by BPSK inputs.
\begin{corollary}
The total $mmse(snr)$ of two users MAC channel with BPSK inputs is given by:
\begin{equation}
mmse(snr)=mmse_1'(snr)+mmse_2'(snr)-\psi(snr)
\end{equation}
\end{corollary}
Where,
\begin{equation}
mmse_1'(snr)=mmse_1(snr)
\end{equation}
and,
\begin{equation}
mmse_2'(snr)=mmse_2(snr)+\psi(snr)
\end{equation}
\begin{proof}
See Appendix B, part II.
\end{proof}
We shall now capitalize on the unveiled connection between the mutual information and the MMSE plus the covariance or cross correlation of the input estimates. In the specfic case in \eqref{simplemodel}, and when both user inputs are decoded jointly, such covariance terms can be easily shown to be equal, over all permutations of the inputs. When the inputs are orthogonal or time-sharing, the covariance terms vanishes, i.e. $\psi(snr)= 0$. In turn, the joint mutual information is just the sum of the rates of both inputs or the integral of the MMSE of both users. However, a general form of the joint mutual information that clarifies the new fundamental relation with unequal covariances is when $\psi(snr)\neq 0$. This is given by the following corollary.
\begin{corollary}
The total $I(snr)$ of two users MAC channel with BPSK inputs is given by:
\begin{equation}
I(snr)=I_1'(snr)+I_2'(snr)
\end{equation}
\end{corollary}
$I_1'(snr)$ corresponds to the mutual information of user 1 given the other user is considered as noise and $I_2'(snr)$ is the mutual information of user 2 given that user 1 is decoded first. 
\begin{proof}
See Appendix B, part II.
\end{proof}
Figure 1 illustrates the mutual information per user in a MAC and the sum rates under equivalent powers and compared to the case of two users over SISO parallel channels. Its quite clear now, why the mutual information for a MAC Gaussian channel approaches 1.5 bits/sec/Hz when both inputs have similar power, incuring 0.5 bits/sec/Hz loss, as previously explained in \cite{newSamah}, and why it doesnt approach the one of parallel Gaussian channels unless unblanced power allocation takes place -the so called mercury/waterfilling, which approaches 2 bits/sec/Hz for BPSK at high SNRs. Moreover, when successive decoding takes place, an ufair rate allocation takes place, were the user decoded first will pay the price from his achievable rates. This can be also well explained in terms of the MMSE, where the MMSE of the user decoded first has a scaled SNR, with a scaling factor less than one. This will let this scaled MMSE not to decay to zero, however, it saturates at high SNR to a point above the zero, at 0.5 for this example.
\begin{figure}[ht!]
    \begin{center}
        \mbox{\includegraphics[width=3.9in,height=2.9in]{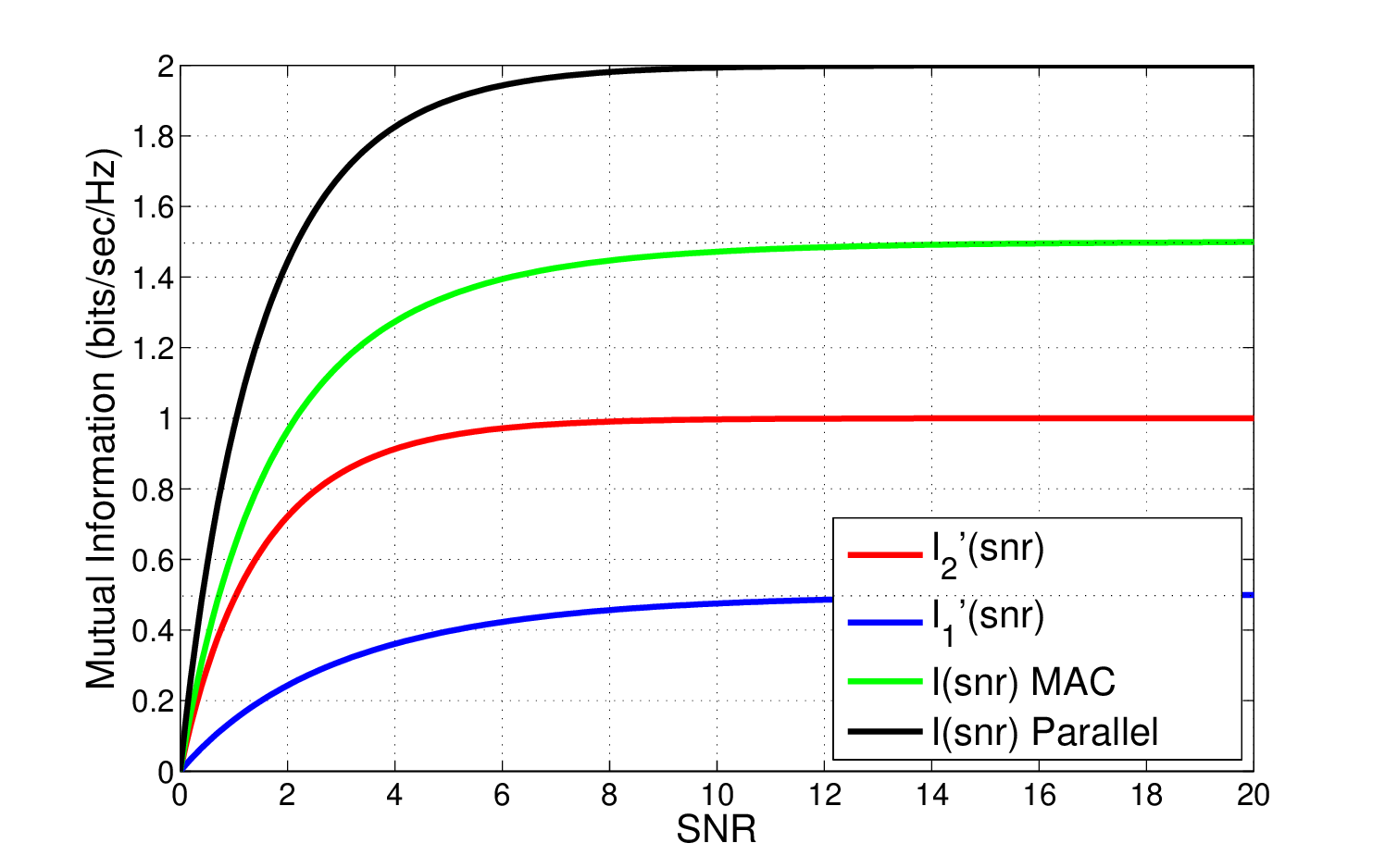}}
    \vspace{-1cm}
    \caption{The two-user per MAC rates and sum rates with BPSK inputs.}
    \label{fig:Figure11}
    \end{center}
    \label{Figure11}
		\vspace{-0.1cm}
 \end{figure}
\section{Multiuser I-MMSE in the Low-SNR Regime}
We now consider the two-user MAC Gaussian channel with arbitrary input distributions in the regime of low-snr. Consider a zero-mean uncorrelated complex inputs, with $\bf{\mathbb{E}[x_1x_1^{\dag}]=I}$, 
$\mathbb{E}[\bf{x_2x_2^{\dag}}]=I$, ~$\mathbb{E}[{\bf{x_1x_1}}^{T}]=0$, and $\mathbb{E}[{\bf{x_2x_2}}^{T}]=0$. We consider the low-snr expansion to the MMSE of equation \eqref{14}. Note that it can be easily deduced that the Taylor expansion of the non$-$linear MMSE in \eqref{14} will lead to the first order Taylor expansion of the linear MMSE for the Gaussian inputs setup. Thus, the low-snr expansion of the MMSE matrix can be expressed as:
\vspace{-0.1cm}
\small
\begin{multline}
\label{eq.15}
{\bf{E}}={\bf{I}}-{({\bf{H_1P_1}})}^{\dag}{\bf{H_1P_1}}.{snr}\\
-{\bf{(H_2P_2)^{\dag}H_2P_2}}.{snr}+ \mathcal{O}(snr^2),
\end{multline}
\normalsize
with $\bf{E=E_1+E_2}$. Consequently,
\small
\begin{multline}
mmse(snr)=Tr\left\{\bf{H_1P_1}\bf{E_1}{(\bf{H_1P_1})}^{\dag}\right\} \\
+Tr \left\{\bf{H_2P_2}\bf{E_2}{(\bf{H_2P_2})}^{\dag}\right\} 
\end{multline}
	\vspace{-0.35cm}
\begin{multline}
~~=Tr \left\{\bf{H_1P_1}{(\bf{H_1P_1})}^{\dag}\right\} \\
+Tr \left\{\bf{H_2P_2}{(\bf{H_2P_2})}^{\dag}\right\} \\
- Tr \left\{(\bf{H_1P_1}{(\bf{H_1P_1})}^{\dag})^{2}\right\}.snr \\
- Tr\left\{(\bf{H_2P_2}{(\bf{H_2P_2})}^{\dag})^{2}\right\}.snr + \mathcal{O}(snr^2).
\end{multline}
\normalsize
Note that due to our new result of Theorem~\ref{theorem1.0}, we cannot apply immediately the fundamental relationship between mutual information and MMSE in \cite{35}, \cite{34}. Therefore, the low-snr expansion of $\psi(snr)$, the covarince or the cross correlation between the inputs estimates can be expressed as:
\vspace{-0.1cm}
\small
\begin{multline}
\label{lowsnr}
\psi(snr)=- Tr\left\{\bf{H_1P_1}{(\bf{H_1P_1})}^{\dag}\bf{H_2P_2}{(\bf{H_2P_2})}^{\dag}\right\}.{snr} \\
-Tr\left\{\bf{H_2P_2}{(\bf{H_2P_2})}^{\dag}\bf{H_1P_1}{(\bf{H_1P_1})}^{\dag}\right\}.{snr}
\end{multline}
\normalsize
\vspace{-0.1cm}
Applying our new result, the low-snr Taylor expansion of the mutual information is given in the following theorem.
\begin{theorem}
\label{LOWSNR}
The low-snr Taylor expansion of the mutual information of the two user MAC is given by:
\vspace{-0.1cm}
\small
\begin{multline}
\label{lowsnr}
I(snr)=Tr \left\{\bf{H_1P_1}{(\bf{H_1P_1})}^{\dag}\right\}.snr + \\
Tr \left\{\bf{H_2P_2}{(\bf{H_2P_2})}^{\dag}\right\}.snr \\
- Tr\left\{(\bf{H_1P_1}{(\bf{H_1P_1})}^{\dag})^{2}\right\}.{snr^2} \\
- Tr\left\{(\bf{H_2P_2}{(\bf{H_2P_2})}^{\dag})^{2}\right\}.{snr^2} \\
- Tr\left\{\bf{H_1P_1}{(\bf{H_1P_1})}^{\dag}\bf{H_2P_2}{(\bf{H_2P_2})}^{\dag}\right\}.{snr^2} \\
-Tr\left\{\bf{H_2P_2}{(\bf{H_2P_2})}^{\dag}\bf{H_1P_1}{(\bf{H_1P_1})}^{\dag}\right\}.{snr^2}+\mathcal{O}(snr^3)
\end{multline}
\end{theorem}
\normalsize
\begin{proof}
See Appendix C
\end{proof}
The wideband slope $-$ which indicates how fast the capacity is achieved in terms of required bandwidth $-$ is inversely proportional to the second order terms of the mutual information in the low-snr Taylor expansion \eqref{lowsnr}. Therefore, this term is a key low-power performance measure since the bandwidth required to sustain a given rate with a given low power, i.e., minimal energy per bit, is inversely proportional to this term \cite{93}. Further, its clear that the 5th and 6th term in \eqref{lowsnr} are due to $\psi(snr)$, to which they play a fundamental role in the rate losses encountered at the low-snr regime.
\section{Conclusions}
We generlize the fundamental relation between the mutual information and the MMSE, the I-MMSE identity in all its current forms, to a new fundamental relation, the Multiuser I-MMSE, which applies to multiuser vector and scalar channel setups. Further, we proved our generlization by deriving the relation for the joint mutual information, conditional and non-conditional mutual information. We capitalize on our unveiled generalized relation to find explicit closed forms of the mutual information and the MMSE of multiuser channels driven by BPSK inputs, and to derive the mutiuser I-MMSE at the regime of low SNR. Besides, the impact of the result is many fold. We mainly quantify the data rate losses due to the interference, which constitutes the gap from the cut-set upper bound. This gap $\int\psi(snr)snr$ in MAC channels as a special case of interference channels is characterized with respect to the users channels, precoding (power allocation) and the decoding or estimation process of inputs for the first time. In turn, we allow for future characterization of the capacity of interference channels. Additionaly, this result allows for future investigation and characterization of the network I-MMSE. On the other hand, the new fundamental relation will have high impact on future designs of transmission schemes that are interference-aware, due to the awareness of the covariance (gap) introduced due to the interference. It will also have impact on statistical signal processing applications that are based on classification of mixtures of data in a measurement system. 
\section{Appendix A: Proof of Theorem~\ref{theorem1.0}}
The conditional probability density for the two-user MAC can be written as follows:
\small
\begin{equation}
p_{y|x_1,x_2}({\bf{y|x_1, x_2}})=\frac{1}{\pi^{n_r}}e^{-\left\|y-\sqrt{snr}{\bf{{H_1P_1x_1}}}-\sqrt{snr}{\bf{H_2P_2x_2}}\right\|^{2}}
\end{equation}
\normalsize
Thus, the corresponding mutual information is:
\begin{equation}
I({\bf{x_1, x_2;y}})=\mathbb{E}\left[log \left(\frac{p_{y|x_1,x_2}({\bf{y|x_1, x_2}})}{p_{y}({\bf{y}})}\right)\right]
\end{equation}
\begin{equation}
I({\bf{x_1, x_2;y}})=-n_r{{log(\pi e)}} -\mathbb{E}\left[log\left(p_{y}({\bf{y}})\right)\right]
\end{equation}
\begin{equation}
I({\bf{x_1, x_2;y}})=-n_r{{log(\pi e)}} -  \int p_{y}({\bf{y}})log\left(p_{y}({\bf{y}})\right) d\bf{y}
\end{equation}
Then, the derivative of the mutual information with respect to the SNR is as follows:
\begin{equation}
\frac{d I({\bf{x_1, x_2;y}})}{d snr}=-\frac{\partial}{\partial snr} \int {p_{y}}({\bf{y}})log \left(p_{y}({\bf{y}})\right) d\bf{y}
\end{equation}
\begin{equation}
~~=-\int \left({p_{y}}({\bf{y}})\frac{1}{{p_{y}}({\bf{y}})}+log \left(p_{y}({\bf{y}})\right)\right)\frac{\partial {p_{y}}({\bf{y}})}{\partial snr} d\bf{y}
\end{equation}
\begin{equation}
\label{45.0}
~~=-\int \left(1+log\left(p_{y}({\bf{y}})\right)\right)\frac{\partial {p_{y}}({\bf{y}})}{\partial snr} d\bf{y}
\end{equation}
Where the probability density function of the received vector $\bf{y}$ is given by:
\begin{equation}
p_{y}({\bf{y}})=\sum_{{\bf{x_{1},x_{2}}}} p_{y|x_{1},x_{2}}({\bf{y|x_{1},x_{2}}})p_{x_1,x_2}({\bf{x_{1}}},{\bf{x_{2}}})
\end{equation}
\begin{equation}
~~=\mathbb{E}_{x_1,x_2}\left[{p_{y|x_{1},x_{2}}}({\bf{y|x_{1},x_{2}}})\right]
\end{equation}
The derivative of the conditional output with respect to the SNR can be written as:
\scriptsize
\begin{multline}
\frac{\partial {{p_{y|x_{1},x_{2}}}}({\bf{y|x_{1},x_{2}}})}{\partial snr}=\\
-{{p_{y|x_{1},x_{2}}}}({\bf{y|x_{1},x_{2}}})\frac{\partial }{\partial snr}\left({\bf{y}}-\sqrt{snr}{\bf{H_1P_1x_1}}-\sqrt{snr}{\bf{H_2P_2x_2}}\right)^{\dag} \times \\
\left({\bf{y}}-\sqrt{snr}{\bf{H_1P_1x_1}}-\sqrt{snr}{\bf{H_2P_2x_2}}\right)
\end{multline}
\vspace{-0.2cm}
\begin{multline}
~~=-\frac{1}{\sqrt{snr}}\left(({\bf{H_1P_1x_1}})^\dag + ({\bf{H_2P_2x_2}})^\dag \right) \times \nonumber\\
 \left({\bf{y}}-\sqrt{snr}{\bf{H_1P_1x_1}}-\sqrt{snr}{\bf{H_2P_2x_2}}\right) \times \nonumber\\
{p_{y|x_{1},x_{2}}}({\bf{y|x_{1},x_{2}}})
\end{multline}
\vspace{-0.1cm}
\begin{equation}
~~=-\frac{1}{\sqrt{snr}}\left(({\bf{H_1P_1x_1}})^\dag + ({\bf{H_2P_2x_2}})^\dag \right){\bf{\nabla_{y}}}{p_{y|x_{1},x_{2}}}({\bf{y|x_{1},x_{2}}})
\end{equation}
\normalsize
Therefore, we have:
\scriptsize
\begin{multline}
\label{46.0}
\mathbb{E}_{x_1,x_2}\left[{\nabla_{snr}}{p_{y|x_{1},x_{2}}}({\bf{y|x_{1},x_{2}}})\right]=\\
\mathbb{E}_{x_1,x_2}\left[-\frac{1}{\sqrt{snr}}\left(({\bf{H_1P_1x_1}})^\dag + ({\bf{H_2P_2x_2}})^\dag \right)\nabla_{\bf{y}}p_{y|x_{1},x_{2}}({\bf{y|x_{1},x_{2}}})\right]
\end{multline}
\normalsize
Substitute \eqref{46.0} into \eqref{45.0}, we get:
\scriptsize
\begin{multline}
\frac{d I({\bf{x_1, x_2;y}})}{d snr}=\frac{1}{\sqrt{snr}} \int \left(1+{log}\left({p_{y}}({\bf{y}})\right)\right) \times \\
\mathbb{E}_{x_1,x_2}[\left(({\bf{H_1P_1x_1}})^\dag+({\bf{H_2P_2x_2}})^\dag \right) \times \\
\bf{\nabla_{y}}\bf{p_{y|x_{1},x_{2}}}(\bf{y|x_{1},x_{2}})] d\bf{y}
\end{multline}
\vspace{-0.2cm}
\begin{multline}
~~=\frac{1}{\sqrt{snr}}\mathbb{E}_{x_1,x_2}[(\int \left(1+{log}\left({p_{y}}({\bf{y}})\right)\right)\left(({\bf{H_1P_1x_1}})^\dag+({\bf{H_2P_2x_2}})^\dag \right) \times \\
\nabla_{\bf{y}}{p_{y|x_{1},x_{2}}}({\bf{y|x_{1},x_{2}}})d\bf{y}) ]
\end{multline}
\normalsize
Using integration by parts applied to the real and imaginary parts of $\bf{y}$ we have:
\scriptsize
\begin{multline}
\label{47.0}
\int \left(1+{log}\left({p_{y}}({\bf{y}})\right)\right) \frac{\partial p_{y|x_{1},x_{2}}({\bf{y|x_{1},x_{2}}})}{\partial t} d{\bf{t}}= \\
\int \left(1+{log}\left({p_{y}}({\bf{y}})\right)\right) p_{y|x_{1},x_{2}}({\bf{y|x_{1},x_{2}}}){|}_{-\infty}^{\infty} \\
-\int_{-\infty}^{\infty} \frac{1}{p_{y}({\bf{y}})} \frac{\partial {p_{y}}({\bf{y}})}{\partial t}{p_{y|x_{1},x_{2}}}({\bf{y|x_{1},x_{2}}})d\bf{t}
\end{multline}
\normalsize
The first term in \eqref{47.0} goes to zero as $\left\|\bf{y}\right\| \rightarrow \infty$. Therefore,
\scriptsize
\begin{multline}
\frac{d I({\bf{x_1, x_2;y}})}{d snr}=\frac{1}{\sqrt{snr}}\mathbb{E}_{x_1,x_2}[-\int ( (({\bf{H_1P_1x_1}})^\dag +({\bf{H_2P_2x_2}})^\dag ) \times \\
\frac{p_{y|x_{1},x_{2}}({\bf{y|x_{1},x_{2}}})}{p_{y}({\bf{y}})}\nabla_{\bf{y}}p_{y}({\bf{y}})d\bf{y} ) ]
\end{multline}
\begin{multline}
\frac{d I({\bf{x_1, x_2;y}})}{d snr}=-\frac{1}{\sqrt{snr}}\int {\nabla_{\bf{y}}}{p_{y}({\bf{y}})} \times \\
\mathbb{E}_{x_1,x_2} [(({\bf{H_1P_1x_1}})^\dag+({\bf{H_2P_2x_2}})^\dag )
\frac{{p_{y|x_{1},x_{2}}}({\bf{y|x_{1},x_{2}}})}{{p_{y}({\bf{y}})}}]d\bf{y}
\end{multline}
\vspace{-0.2cm}
\begin{multline}
\label{48.0}
\frac{d I({\bf{x_1, x_2;y}})}{d snr}=-\frac{1}{\sqrt{snr}}\int \nabla_{\bf{y}}{{p_{y}({\bf{y}})}}\mathbb{E}_{x_1,x_2}[({\bf{H_1P_1}})^\dag \mathbb{E}_{x_1|y}\left[{\bf{x_1|y}}\right]^{\dag} \\
+({\bf{H_2P_2}})^\dag  \mathbb{E}_{x_2|y}\left[{\bf{x_2|y}}\right]^{\dag} ] d\bf{y}
\end{multline}
\normalsize
However,
\scriptsize
\begin{multline}
\label{49.0}
\nabla_{\bf{y}}{p_{y}(\bf{y})}=\nabla_{\bf{y}}\mathbb{E}_{x_1,x_2}\left[{p_{y|x_{1},x_{2}}}({\bf{y|x_{1},x_{2}}})\right] \\
=\mathbb{E}_{x_1,x_2}\left[\nabla_{\bf{y}}{p_{y|x_{1},x_{2}}}({\bf{y|x_{1},x_{2}}})\right] \\
=-\mathbb{E}_{x_1,x_2}\left[{p_{y|x_{1},x_{2}}}({\bf{y|x_{1},x_{2}}})\left({\bf{y}}-\sqrt{snr}{\bf{H_1P_1x_1}}-\sqrt{snr}{\bf{H_2P_2x_2}}\right)\right]\\
=-\mathbb{E}_{x_1,x_2}\left[{p_{y}}({\bf{y}})\left({\bf{y}}-\sqrt{snr}{\bf{H_1P_1x_1}}-\sqrt{snr}{\bf{H_2P_2x_2}}\right)|{\bf{y}}\right] \\
=-{p_{y}}({\bf{y}})\left({\bf{y}}-\sqrt{snr}{\bf{H_1P_1}}\mathbb{E}_{x_1|y}[{\bf{x_1|y}}]-\sqrt{snr}{\bf{H_2P_2}}\mathbb{E}_{x_2|y}[{\bf{x_2|y}}]\right)\\
\end{multline}
\normalsize
\vspace{-0.2cm}
Substitute \eqref{49.0} into \eqref{48.0} we get:
\scriptsize
\begin{multline}
\frac{d I({\bf{x_1, x_2;y}})}{d snr}=\frac{1}{\sqrt{snr}}\int p_{y}({\bf{y}})({\bf{y}}-\sqrt{snr}{\bf{H_1P_1}}\mathbb{E}_{x_1|y}[{\bf{x_1|y}}]\\
-\sqrt{snr}{\bf{H_2P_2}} \mathbb{E}_{x_2|y}[{\bf{x_2|y}}]) \times \nonumber \\
\mathbb{E}_{x_1,x_2}\left(({\bf{H_1P_1}})^\dag \mathbb{E}_{x_1|y}\left[{\bf{x_1|y}}\right]^{\dag}+({\bf{H_2P_2}})^\dag \mathbb{E}_{x_2|y}\left[{\bf{x_2|y}}\right]^{\dag} \right) d\bf{y}
\end{multline}
\vspace{-0.2cm}
\begin{multline}
\frac{d {\bf{I(x_1, x_2;y)}}}{d snr}=\frac{1}{{\sqrt{snr}}}\mathbb{E}_y[\bf{yx_1^{\dag}}]{\bf{(H_1P_1)^{\dag}}}+\frac{1}{\sqrt{snr}}\mathbb{E}_y[\bf{yx_2^{\dag}}]\bf{(H_2P_2)^{\dag}} \nonumber \\
-\mathbb{E}_y[{\bf{H_1P_1}}\mathbb{E}_{x_1|y}[{\bf{x_1|y}}]\mathbb{E}_{x_1|y}[{\bf{x_1|y}}]^{\dag}]{\bf{(H_1P_1)^{\dag}}} \nonumber \\
-\mathbb{E}_y[{\bf{H_1P_1}}\mathbb{E}_{x_1|y}[{\bf{x_1|y}}]\mathbb{E}_{x_2|y}[{\bf{x_2|y}}]^{\dag}]{\bf{(H_2P_2)^{\dag}}} \nonumber \\
-\mathbb{E}_y[{\bf{H_2P_2}}\mathbb{E}_{x_2|y}[{\bf{x_2|y}}]\mathbb{E}_{x_2|y}[{\bf{x_2|y}}]^{\dag}]{\bf{(H_2P_2)^{\dag}}} \nonumber \\
-\mathbb{E}_y[{\bf{H_2P_2}}\mathbb{E}_{x_2|y}[{\bf{x_2|y}}]\mathbb{E}_{x_1|y}[{\bf{x_1|y}}]^{\dag}]{\bf{(H_1P_1)^{\dag}}}
\end{multline}
\normalsize
Therefore,
\scriptsize
\begin{multline}
\frac{d I({\bf{x_1, x_2;y}})}{d snr}={\bf{H_1P_1}}\mathbb{E}_{x_1}[{\bf{x_1x_1^{\dag}}}]{\bf{(H_1P_1)^{\dag}}} \nonumber \\
-{\bf{H_1P_1}}\mathbb{E}_y[\mathbb{E}_{x_1|y}[{\bf{x_1|y}}]\mathbb{E}_{x_1|y}[{\bf{x_1|y}}]^{\dag}{\bf{(H_1P_1)^{\dag}}} \nonumber \\
+{\bf{H_2P_2}}\mathbb{E}_{x_2}[{\bf{x_2x_2^{\dag}}}]{\bf{(H_2P_2)^{\dag}}} \nonumber \\
-{\bf{H_2P_2}}\mathbb{E}_y[\mathbb{E}_{x_2|y}[{\bf{x_2|y}}]\mathbb{E}_{x_2|y}[{\bf{x_2|y}}]^{\dag}{\bf{(H_2P_2)^{\dag}}}  \nonumber \\
-{\bf{H_1P_1}}\mathbb{E}_y[\mathbb{E}_{x_1|y}[\bf{x_1|y}]\mathbb{E}_{x_2|y}[{\bf{x_2|y}}]^{\dag}]{\bf{(H_2P_2)^{\dag}}}\nonumber \\
-{\bf{H_2P_2}}\mathbb{E}_y[\mathbb{E}_{x_2|y}[{\bf{x_2|y}}]\mathbb{E}_{x_1|y}[{\bf{x_1|y}}]^{\dag}]{\bf{(H_1P_1)^{\dag}}}
\end{multline}
\vspace{-0.2cm}
\begin{multline}
\frac{d I({\bf{x_1, x_2;y}})}{d snr}={\bf{H_1P_1}}{\bf{E_1}}{\bf{(H_1P_1)^{\dag}}}+{\bf{H_2P_2}}{\bf{E_2}}{\bf{(H_2P_2)^{\dag}}} \nonumber \\
-{\bf{H_1P_1}}\mathbb{E}_y[\mathbb{E}_{x_1|y}[{\bf{x_1|y}}]\mathbb{E}_{x_2|y}[{\bf{x_2|y}}]^{\dag}]{\bf{(H_2P_2)^{\dag}}} \nonumber \\
-{\bf{H_2P_2}}\mathbb{E}_y[\mathbb{E}_{x_2|y}[{\bf{x_2|y}}]\mathbb{E}_{x_1|y}[{\bf{x_1|y}}]^{\dag}]{\bf{(H_1P_1)^{\dag}}}
\end{multline}
\normalsize
Therefore, the derivative of the mutual information with respect to the SNR and the per users mmse and input estimates (or covariances) is as follows:
\scriptsize
\begin{multline}
\frac{d I({\bf{x_1, x_2;y}})}{d snr}=mmse_1(snr)+mmse_2(snr)\\
-Tr\left\{{\bf{H_1P_1}}\mathbb{E}_{y}[{\bf{\widehat{x}_1{\widehat{x}_2}^{\dag}}}]{\bf{(H_2P_2)^{\dag}}}\right\}  \nonumber \\
-Tr\left\{{\bf{H_2P_2}}\mathbb{E}_y[{\bf{\widehat{x}_2{\widehat{x}_1}^{\dag}}}]{\bf{(H_1P_1)^{\dag}}}\right\}
\end{multline}
\normalsize
\vspace{-0.1cm}
Therefore, we can write the derivative of the derivative of the mutual information with respect to the snr as follows:
\vspace{-0.1cm}
\small
\begin{equation}
\frac{d I(snr)}{d snr}=mmse(snr)+\psi(snr)
\end{equation}
\normalsize
\vspace{-0.05cm}
Therefore, Theorem~\ref{theorem1.0} has been proved as a generalization of the one by Guo, Shamai, Verdu in \cite{34} to the multiuser case.
\section{Appendix B: Multiuser MMSE(snr) and I(snr) for BPSK Inputs}
\subsection{Part I} 
Consider the simplified case for a channel model given by:
\begin{equation}
\label{simplemodel}
{\bf{y}}=\sqrt{snr}~{\bf{x_1}}+\sqrt{snr}~\bf{x_2}+\bf{n},
\end{equation}
The total MMSE is given as:
\begin{equation}
mmse(snr)=mmse_1(snr)+mmse_2(snr)
\end{equation}
Therefore, we can write the non-linear MMSE matrix for each user respectively as:
\small
\begin{multline}
\label{5.5.0}
\bf{E_1}=\bf{\mathbb{E}[(\bf{x_1}-\bf{\mathbb{E}[x_1|y]})(\bf{x_1}-\bf{\mathbb{E}[x_1|y]})^{\dag}]}\\
=\mathbb{E}[\bf{x_1x_1^{\dag}}]-\mathbb{E}[\bf{\mathbb{E}[x_1|y]\mathbb{E}[x_1|y]}],
\end{multline}
\begin{multline}
\label{5.5.1}
\bf{E_2}=\bf{\mathbb{E}[(\bf{x_2}-\bf{\mathbb{E}[x_2|y]})(\bf{x_2}-\bf{\mathbb{E}[x_2|y]})^{\dag}]}\\
=\mathbb{E}[\bf{x_2x_2^{\dag}}]-\mathbb{E}[\bf{\mathbb{E}[x_2|y]\mathbb{E}[x_2|y]}],
\end{multline}
\normalsize
with,
\begin{equation}
\mathbb{E}[{\bf{x_1|y}}]=\frac{\sum_{{\bf{x_1,x_2}}}{\bf{x_1}}p_{y|x_1,x_2}({\bf{y|x_1, x_2}})p_{x_1}({\bf{x_1}})p_{x_2}({\bf{x_2}})}{p_{y}({\bf{y}})}
\end{equation}
\begin{equation}
\label{5.5.01}
~~= \frac{\sum_{{\bf{x_1,x_2}}}{\bf{x_1}}p_{y|x_1,x_2}({\bf{y|x_1,x_2}})p_{x_1}({\bf{x_1}})p_{x_2}({\bf{x_2}})}{\sum_{{\bf{x_1',x_2'}}}p_{{\bf{y|x_1',x_2'}}}({\bf{y|x_1',x_2'}})p_{x_1'}({\bf{x_1'}})p_{x_2'}({\bf{x_2'}})}
\end{equation}
\begin{equation}
\mathbb{E}[{\bf{x_2|y}}]=\frac{\sum_{{\bf{x_1,x_2}}}{\bf{x_2}}p_{y|x_1,x_2}({\bf{y|x_1,x_2}})p_{x_1}({\bf{x_1}})p_{x_2}({\bf{x_2}})}{p_{y}({\bf{y}})}
\end{equation}
\begin{equation} 
\label{5.5.02}
~~= \frac{\sum_{{\bf{x_1,x_2}}}{\bf{x_2}}p_{y|x_1,x_2}({\bf{y|x_1,x_2}})p_{x_1}({\bf{x_1}})p_{x_2}({\bf{x_2}})}{\sum_{{\bf{x_1',x_2'}}}p_{{\bf{y|x_1',x_2'}}}({\bf{y|x_1',x_2'}})p_{x_1'}({\bf{x_1'}})p_{x_2'}({\bf{x_2'}})}
\end{equation}
Its of particlar importance to notice that the conditioning over $x_2$ in $p_{y|x_1,x_2}({\bf{y|x_1, x_2}})$ inside $\mathbb{E}[{\bf{x_1|y}}]$ does not correspond to knowledge of the message $x_2$, but to considering it as noise in this setup. Therefore, we yet account for $p_{x_2}({\bf{x_2}})$. On the other hand, its worth to note also that the conditioning over $x_1$ in $p_{y|x_1,x_2}({\bf{y|x_1, x_2}})$ inside $\mathbb{E}[{\bf{x_2|y}}]$ is just for clarity. However, if $x_1$ is decoded first, then when decoding (estimating) $x_2$ next, we remove $x_1$. Therefore, $\mathbb{E}[{\bf{x_2|y}}]$ is in fact equal to $\mathbb{E}[{\bf{x_2|y-x_1}}]$, this will make $p_{y|x_1,x_2}({\bf{y|x_1, x_2}})$ equals to $p_{y|x_2}({\bf{y|x_2}})$ and $p_{x_1}({\bf{x_1}})$ will be absent accordingly from the equation. 
\normalsize

For the two user MAC driven by BPSK inputs, the values of ${\bf{x_1}}=\{1,-1\}$ and ${\bf{x_2}}=\{1,-1\}$. The non-linear estimates in \eqref{5.5.01} and \eqref{5.5.02} consider that both user inputs are decoded jointly. However, we are interested in successive decoding of the users inputs. 

Therefore, the non-linear estimate of user 1 decoded first in the MAC considering user 2 as noise that scales the SNR of user 1, and so we can write it with respect to all possible permutations of the possible inputs of user 1 as follows:

\vspace{-0.05cm}

\scriptsize
\begin{equation}
\label{5.5.011}
\mathbb{E}[{\bf{x_1|y}}]= \frac{\sum_{{\bf{x_1}}}{\bf{x_1}}p_{y|x_1,x_2}({\bf{y|x_1,x_2}})p_{x_1}({\bf{x_1}})p_{x_2}({\bf{x_2}})}{\sum_{{\bf{x_1'}}}p_{{\bf{y|x_1'}}}({\bf{y|x_1',x_2}})p_{x_1'}({\bf{x_1'}})p_{x_2}({\bf{x_2}})}
\end{equation}
\begin{equation}
\label{5.5.0110}
\mathbb{E}[{\bf{x_1|y}}]=\frac{e^{\frac{-({\bf{y}}-\sqrt{snr})^{2}}{4}}-e^{\frac{-({\bf{y}}+\sqrt{snr})^{2}}{4}}}{e^{\frac{-({\bf{y}}-\sqrt{snr})^{2}}{4}}+e^{\frac{-({\bf{y}}+\sqrt{snr})^{2}}{4}}}
\end{equation}
\normalsize
However,
\scriptsize
\begin{multline}
\mathbb{E}\left[\mathbb{E}({\bf{x_1|y}}){\mathbb{E}({\bf{x_1|y}})^{\dag}}\right] \\
=\int{\left(\frac{\sum_{{\bf{x_1}}}{{\bf{x_1}}}p_{y|x_1,x_2}({\bf{y|x_1,x_2}})p_{x_1}({\bf{x_1}})p_{x_2}({\bf{x_2}})}{p_{y}({\bf{y}})}\right)^{2}p_{y}({\bf{y}}) \bf{dy}}
\end{multline}
\begin{multline}
\label{5511}
\mathbb{E}[\mathbb{E}({\bf{x_1|y}})\mathbb{E}({\bf{x_1|y}})^{\dag}]=\\
\frac{1}{8\pi}\int{\frac{e^{\frac{-({\bf{y}}-\sqrt{snr})^{2}}{4}}-e^{\frac{-({\bf{y}}+\sqrt{snr})^{2}}{4}}}{e^{\frac{-({\bf{y}}-\sqrt{snr})^{2}}{4}}+e^{\frac{-({\bf{y}}+\sqrt{snr})^{2}}{4}}}(e^{\frac{-({\bf{y}}-\sqrt{snr})^{2}}{4}}-e^{\frac{-({\bf{y}}+\sqrt{snr})^{2}}{4}})\bf{dy}}
\end{multline}
\normalsize
\vspace{-0.1cm}
Digging into the depth of the right hand side of equation~\eqref{5511}, we have:
\begin{equation}
({\bf{y}}-\sqrt{snr})^{2}={\bf{y}^2}-{2\sqrt{snr}{\bf{y}}}+{snr}
\end{equation}
\begin{equation}
({\bf{y}}+\sqrt{snr})^{2}={\bf{y^2}}+{2\sqrt{snr}{\bf{y}}}+{snr}
\end{equation}
Thus,
\begin{equation}
\frac{e^{\frac{-({\bf{y}}-\sqrt{snr})^{2}}{4}}-e^{\frac{-({\bf{y}}+\sqrt{snr})^{2}}{4}}}{e^{\frac{-({\bf{y}}-\sqrt{snr})^{2}}{4}}+e^{\frac{-({\bf{y}}+\sqrt{snr})^{2}}{4}}}
=\frac{e^{\frac{\sqrt{snr}}{2}{\bf{y}}}-e^{-\frac{\sqrt{snr}}{2}{\bf{y}}}}{e^{{\frac{\sqrt{snr}}{2}{\bf{y}}}}+e^{-{\frac{\sqrt{snr}}{2}{\bf{y}}}}}
\end{equation}
\begin{equation}
~~=tanh\left(\frac{\sqrt{snr}}{2}\mathcal{R}(\bf{y})\right)
\end{equation}
It follows that:
\scriptsize
\begin{multline}
\mathbb{E}\left[\bf{\mathbb{E}({\bf{x_1|y}})}\bf{\mathbb{E}({\bf{x_1|y}})^{\dag}}\right] \\
=\frac{1}{8\pi}\int{tanh\left(\frac{\sqrt{snr}}{2}\mathcal{R}({\bf{y}})\right)\left(e^{\frac{-({\bf{y}}-\sqrt{snr})^{2}}{4}}-e^{\frac{-({\bf{y}}+\sqrt{snr})^{2}}{4}}\right)d\bf{y}}
\end{multline}
\normalsize
Therefore,
\scriptsize
\begin{multline}
\mathbb{E}\left[\bf{\mathbb{E}({\bf{x_1|y}})}{\mathbb{E}({\bf{x_1|y}})^{\dag}}\right]\\
=\frac{1}{8\pi}\int_{y\in \textsl{C}}{tanh\left(\frac{\sqrt{{snr}}}{2}{\mathcal{R}(y)}\right)e^{\frac{-\left({\bf{y}}-{\sqrt{snr}}\right)^{2}}{4}}}\\
-\frac{1}{8\pi}\int_{y\in \textsl{C}}{tanh\left(\frac{\sqrt{{snr}}}{2}{\mathcal{R}(y)}\right)e^{\frac{-\left({\bf{y}}+{\sqrt{snr}}\right)^{2}}{4}}d\bf{y}}
\end{multline}
\normalsize
However, it is known that:
\begin{equation}
tanh(-x)=-tanh(x)
\end{equation}
and the expectation remains the same if ${\bf{y}}\sim$ $\mathcal{N}$ $(\sqrt{snr},1)$ replaced by ${\bf{y}}\sim$ $\mathcal{N}$ $(-\sqrt{snr},1)$, due to symmetry, therefore, we have:
\scriptsize
\begin{multline}
\label{5.5.3}
\mathbb{E}\left[\bf{\mathbb{E}({\bf{x_1|y}})}\bf{\mathbb{E}({\bf{x_1|y}})^{\dag}}\right]\\
=\frac{1}{4\pi}\int_{y\in \textsl{C}}{tanh\left(\frac{\sqrt{{snr}}}{2}{\bf{y}}\right)e^{\frac{-\left({\bf{y}}-\sqrt{snr}\right)^{2}}{4}}\bf{dy}}
\end{multline}
\normalsize
Therefore, due to marginalization of the complex domain into the real domain, substituting $\mathbb{E}[{\bf{x_1x_1^{\dag}}}]=1$ into \eqref{5.5.0} the scaled MMSE of user 1 over a MAC channel with BPSK inputs is given by:
\vspace{-0.1cm}
\scriptsize
\begin{multline}
\label{5.5.4}
mmse_1'(snr)=1-\frac{1}{4\sqrt{\pi}}\int_{y \in \textsl{R}}tanh\left(\frac{\sqrt{{snr}}}{2}{\bf{y}}\right) \times \nonumber \\
e^{\frac{-\left({\bf{y}}-\sqrt{snr}\right)^{2}}{4}}d{\bf{y}}
\end{multline}
\normalsize
Therefore, Theorem~\ref{thm6} has been proved.
\subsection{Part II}
Due to the relation between the MMSE and the mutual information for SISO channels, the mutual information for user 1 decoded first and with a $2\sigma^2$ scaled $snr$ is given by:
\vspace{-0.15cm}
\scriptsize
\begin{equation}
I_1'(snr)=\frac{snr}{4}-\frac{1}{4\sqrt{\pi}}\int_{y \in \textsl{R}}log~cosh\left(\frac{\sqrt{{snr}}}{2}{\bf{y}}\right)e^{\frac{-\left({\bf{y}}-\sqrt{snr}\right)^{2}}{4}} d{\bf{y}}
\end{equation}
\normalsize
Where $\sigma^2=2$ is the sum of the noise and interference power variance. Following similar steps to the ones above, user 2 will be decoded next given (conditioned) on the knowledge of user 1 who is decoded first, therefore, the non-linear estimate of user 2 message removing user 1 message is:
\begin{equation}
\label{5.5.01122}
\mathbb{E}[{\bf{x_2|y}}]= \frac{\sum_{{\bf{x_2}}}{\bf{x_2}}p_{y|x_2}({\bf{y|x_2}})p_{x_2}({\bf{x_2}})}{\sum_{{\bf{x_2'}}}p_{{\bf{y|x_2'}}}({\bf{y|x_2}})p_{x_2}({\bf{x_2}})}
\end{equation}
\begin{equation}
\label{5.5.011022}
\mathbb{E}[{\bf{x_2|y}}]=\frac{e^{\frac{-({\bf{y}}-\sqrt{snr})^{2}}{2}}-e^{\frac{-({\bf{y}}+\sqrt{snr})^{2}}{2}}}{e^{\frac{-({\bf{y}}-\sqrt{snr})^{2}}{2}}+e^{\frac{-({\bf{y}}+\sqrt{snr})^{2}}{2}}}
\end{equation}
Following the same steps as before, and capitalizing on the new unvelied relation, the MMSE of user 2 will be given by,
\scriptsize
\begin{multline}
\label{5.5.41}
mmse_2(snr)=1-\frac{1}{\sqrt{2\pi}}\int_{y \in \textsl{R}}tanh\left(\sqrt{{snr}}{\bf{y}}\right)e^{\frac{-\left({\bf{y}}-\sqrt{snr}\right)^{2}}{2}}d{\bf{y}}
\end{multline}
\normalsize
and the mutual information for user 2 decoded next is given by:
\vspace{-0.15cm}
\scriptsize
\begin{multline}
I_2'(snr)={snr}-\frac{1}{\sqrt{2\pi}}\int_{y \in \textsl{R}}log~cosh\left(\sqrt{{snr}}{\bf{y}}\right)e^{\frac{-\left({\bf{y}}-\sqrt{snr}\right)^{2}}{2}} d{\bf{y}}
\end{multline}
\normalsize
Notice that due to the new fundamental relation between mmse(snr) and the mutual information, we can observe the effect of the covariance terms $\psi(snr)$, given as;
\vspace{-0.15cm}
\scriptsize
\begin{multline}
\psi(snr)=-\mathbb{E}_y[\bf{\mathbb{E}_{x_1|y}[\bf{x_1|y}]\mathbb{E}_{x_2|y}[\bf{x_2|y}]^{\dag}}] \\
 -\mathbb{E}_y[\bf{\mathbb{E}_{x_2|y}[\bf{x_2|y}]\mathbb{E}_{x_1|y}[\bf{x_1|y}]^{\dag}}]
\end{multline}
\normalsize

Following similar steps to the ones before, we can see that the covariance term will have a negative value which explains the loss in the mutual information in $I_1'(snr)$, and correspondingly to $I(snr)$. Therefore, the covariances of such setup are given as:
\vspace{-0.15cm}
\small
\begin{multline}
\mathbb{E}_y[{\mathbb{E}_{x_1|y}[{\bf{x_1|y}}]\mathbb{E}_{x_2|y}[{\bf{x_2|y}}]^{\dag}}] \\
=\frac{1}{\sqrt{2\pi}}\int_{y \in \textsl{R}}tanh\left(\frac{\sqrt{{snr}}}{2}{\bf{y}}\right)e^{\frac{-\left({\bf{y}}-\sqrt{snr}\right)^{2}}{2}}d{\bf{y}}
\end{multline}
\normalsize
and,
\vspace{-0.15cm}
\small
\begin{multline}
\mathbb{E}_y[{\mathbb{E}_{x_2|y}[{\bf{x_2|y}}]\mathbb{E}_{x_1|y}[{\bf{x_1|y}}]^{\dag}}] \\
=\frac{1}{4\sqrt{\pi}}\int_{y \in \textsl{R}}tanh\left(\sqrt{{snr}}{\bf{y}}\right)e^{\frac{-\left({\bf{y}}-\sqrt{snr}\right)^{2}}{4}}d{\bf{y}}
\end{multline}
\normalsize
Both terms are not equal, which can be explained by $p_y({\bf{y}})$ that is different in the integration based on who is decoded first. 

The new fundamental relation between the mutual information and the MMSE plus covariance states that,
\begin{equation}
\label{nfr}
\frac{dI(snr)}{dsnr}=mmse_1(snr)+mmse_2+\psi(snr)
\end{equation}
However, we derive the mutual information $I_1'(snr)$ based on the following:
\begin{equation}
\frac{dI_1'(snr)}{dsnr}=mmse_1'(snr)
\end{equation}
and the mutual information $I_2'(snr)$ was derived based on the following:
\begin{equation}
\frac{dI_2'(snr)}{dsnr}=mmse_2(snr)
\end{equation}
It follows from \eqref{nfr} that,
\begin{equation}
\frac{dI_2(snr)}{dsnr}=mmse_2(snr)+\psi(snr)
\end{equation}
Therefore,
\begin{equation}
\frac{dI_2(snr)}{dsnr}-\frac{dI_2'(snr)}{dsnr}=\psi(snr)
\end{equation}
This means that:
\begin{multline}
mmse(snr)=mmse_1(snr)+mmse_2(snr)\\
=mmse_1'(snr)+mmse_2(snr)\\
=mmse_1'(snr)+mmse_2'(snr)-\psi(snr)
\end{multline}
and,
\begin{equation}
I(snr)=I_1'(snr)+I_2'(snr)
\end{equation}
\begin{equation}
I(snr)=I_1'(snr)+I_2(snr)-\int \psi(snr) dsnr
\end{equation}
Moreover, due to the reasons discussed earlier, or when both inputs are decoded jointly, those covariance terms in $\psi(snr)$ might collapse to zero and so the mutual information will be the sum of the integral of both users MMSEs $mmse_1(snr)+mmse_2(snr)$. Therefore, Theorem~\ref{thm7} has been proved with its following corrollaries.
\section{Appendix C: Proof of Theorem~\ref{LOWSNR}}
First we will find the low-snr expansion of the MMSE matrix ${\bf{E_1}}+{\bf{E_2}}$ for user 1 and user 2 in~\eqref{eq.15} as $snr\rightarrow 0$. Therefore, we will first derive the low-snr expansion of the conditional probability exponent given as:
\small
\begin{multline}
|{\bf{y}}-\sqrt{snr}{\bf{H_1P_1x_1}}-\sqrt{snr}{\bf{H_2P_2x_2}}|^2 \nonumber\\
=\left({\bf{y}}-\sqrt{snr}{\bf{H_1P_1x_1}}-\sqrt{snr}{\bf{H_2P_2x_2}}\right)^{\dag} \times \nonumber\\
\left({\bf{y}}-\sqrt{snr}{\bf{H_1P_1x_1}}-\sqrt{snr}{\bf{H_2P_2x_2}}\right) \nonumber\\
=|{\bf{y}}|^2-\sqrt{snr}\left({\bf{y}^{\dag}}{\bf{H_1P_1x_1}}+\left({\bf{y^{\dag}}}{\bf{H_1P_1x_1}}\right)^{\dag}\right) \nonumber\\
+snr({\bf{H_1P_1x_1}})^{\dag}{\bf{H_1P_1x_1}} \nonumber\\
+snr({\bf{H_1P_1x_1}})^{\dag}{\bf{H_2P_2x_2}} \nonumber\\
-\sqrt{snr}\left({\bf{y}^{\dag}}{\bf{H_2P_2x_2}}+\left({\bf{y^{\dag}}}{\bf{H_2P_2x_2}}\right)^{\dag}\right) \nonumber\\
+snr({\bf{H_2P_2x_2}})^{\dag}{\bf{H_2P_2x_2}}+snr({\bf{H_2P_2x_2}})^{\dag}{\bf{H_1P_1x_1}} \nonumber\\
=|{\bf{y}}|^2-2 \sqrt{snr}\mathcal{R}\left({\bf{y^{\dag}}}{\bf{H_1P_1x_1}}\right) \nonumber\\
+ snr|{\bf{H_1P_1x_1}}|^2 +snr({\bf{H_1P_1x_1}})^{\dag}{\bf{H_2P_2x_2}} \nonumber\\
-2 \sqrt{snr}\mathcal{R}\left({\bf{y^{\dag}}}{\bf{H_2P_2x_2}}\right) \nonumber\\
+ snr|{\bf{H_2P_2x_2}}|^2 \nonumber\\
+snr({\bf{H_2P_2x_2}})^{\dag}{\bf{H_1P_1x_1}} \nonumber\\
\end{multline}
\normalsize
Hence,
\footnotesize
\begin{multline}
p_{y|x_1,x_2}\left({\bf{y|x_1,x_2}}\right)= \nonumber \\
\frac{1}{\pi^{n_r}}\exp{\left(-|{\bf{y}}-\sqrt{snr}{\bf{H_1P_1x_1}}-\sqrt{snr}{\bf{H_2P_2x_2}}|^2\right)}
\end{multline}
\footnotesize
\begin{multline}
~~=\frac{1}{\pi^{n_r}}e^{(-|{\bf{y}}|^2)} \times \nonumber\\
\frac{e^{\left(2{\sqrt{snr}}\mathcal{R}({\bf{y^{\dag}H_1P_1x_1}})+2{\sqrt{snr}}\mathcal{R}({\bf{y^{\dag}H_2P_2x_2}})\right)}}
{e^{({snr}\sum\limits_{i=1}^2|{\bf{H_iP_ix_i}}|^2+snr({\bf{H_2P_2x_2}})^{\dag}{\bf{H_1P_1x_1}}+snr({\bf{H_1P_1x_1}})^{\dag}{\bf{H_2P_2x_2}})}}
\end{multline}
\normalsize
However, due to:
\begin{equation}
\label{lab}
\exp{\left(a{\sqrt{snr}}-b{snr}\right)}=\frac{\exp{\left(a{\sqrt{snr}}\right)}}{\exp{\left(b{snr}\right)}} 
\end{equation}
The Taylor expansion of the numerator and the denominator of \eqref{lab} as $snr\rightarrow 0$, is given as,
\begin{multline}
\exp{\left(a{\sqrt{snr}}-b {snr}\right)}=\frac{1+a{\sqrt{snr}}+\mathcal{O}\left({snr}\right)}{1+\mathcal{O}\left({snr}\right)} \nonumber\\
=1+a{\sqrt{snr}}+\mathcal{O}\left({snr}\right)
\end{multline}
Therefore, the low-snr expansion of the conditional probability distribution of the Gaussian noise is defined as:
\small
\begin{multline}
\label{equationQ}
{p}_{y|x_1,x_2}\left({\bf{y|x_1,x_2}}\right)=  \\
\frac{1}{\pi^{n_r}}\exp{\left(-|\bf{y}|^2\right)} \times  \\
(1+2{\sqrt{snr}}\mathcal{R}\left({\bf{y^{\dag}H_1P_1x_1}}\right)+2{\sqrt{snr}}\mathcal{R}\left({\bf{y^{\dag}H_2P_2x_2}}\right) \\
+\mathcal{O}\left({snr}\right))
\end{multline}
\normalsize
\subsection{Derivation of the Multiuser MMSE at the Low SNR}
\vspace{-0.1cm}
The first term of the MMSE matrix of user 1 $\bf{E_1}$ is $\mathbb{E}[\bf{x_1x_1}^{\dag}]=I$. However, to find the second term of $\bf{E_1}$, $\mathbb{E}[{\bf{x_1|y}}]$ is defined as:
\vspace{-0.1cm}
\begin{equation}
\label{nonlinearestimate1}
\mathbb{E}[{\bf{x_1|y}}]=\frac{\sum_{{\bf{x_1,x_2}}}{\bf{x_1}}p_{y|x_1,x_2}({\bf{y|x_1, x_2}})p_{x_1}({\bf{x_1}})p_{x_2}({\bf{x_2}})}{p_{y}({\bf{y}})}
\end{equation}
We need to substitute \eqref{equationQ} into \eqref{nonlinearestimate1} as follows:
\vspace{-0.1cm}
\footnotesize
\begin{multline}
\mathbb{E}_y[\mathbb{E}_{x_1|y}[x_1|y]\left(\mathbb{E}_{x_1|y}[x_1|y]\right)^{\dag}] \nonumber \\
=\int_{y\in\mathbb{C}^{n_r}}\frac{1}{\pi^{n_r}}\exp{\left(-|y|^2\right)}p_{x_1}\left({\bf{x_1}}\right)p_{x_2}\left({\bf{x_2}}\right) \times \nonumber \\
\frac{\sum_{x_1,x_2} {\bf{x_1}}\left(1+\sum\limits_{i=1}^2 2{\sqrt{snr}}\mathcal{R}\left({\bf{y^{\dag}H_iP_ix_i}}\right)+\mathcal{O}\left(snr\right)\right)}{\sum_{x_1',x_2'}\left(1+\sum\limits_{i=1}^2 2{\sqrt{snr}}\mathcal{R}\left({\bf{y^{\dag}H_iP_ix_i'}}\right)+\mathcal{O}\left(snr\right)\right)} \times \\
(\sum_{x_1,x_2} {\bf{x_1}}\left(1+\sum\limits_{i=1}^2 2{\sqrt{snr}}\mathcal{R}\left({\bf{y^{\dag}H_iP_ix_i}}\right)+\mathcal{O}\left(snr\right)\right))^{\dag}d{\bf{y}}  \normalsize
\end{multline}
\normalsize
Recall that $\mathbb{E}[x_1]=\sum_{x_1} {\bf{x_1}}p_{x_1}({\bf{x_1}})=0$, $\mathbb{E}_{x_1,x_2}[{{\bf{x_1x_1}}^T}]=0$, $\mathbb{E}_{x_1,x_2}[{{\bf{x_1x_2}}^T}]=0$, $\mathbb{E}_{x_1,x_2}[{{\bf{x_2x_1}}^T}]=0$,   
$\mathbb{E}_{x_1,x_2}[{\bf{x_1x_1^\dag}}]={\bf{I}}$, and $\mathbb{E}_{x_1,x_2}[{\bf{x_2x_2^\dag}}]={\bf{I}}$. 

Therefore, the numerator of \eqref{nonlinearestimate1} is given by, 
\vspace{-0.15cm}

\scriptsize
\begin{multline}
\label{denominator}
\exp{(-|{\bf{y}}|^2)} \times \nonumber \\
(\sum_{x_1,x_2} {\bf{x_1}} (1+2\sqrt{snr}\mathcal{R}({\bf{y^{\dag}}} {\bf{H_1P_1x_1}}))p_{x_1}({\bf{x_1}})p_{x_2}({\bf{x_2}}) \nonumber \\
+2\sqrt{snr}\mathcal{R}({\bf{y^{\dag}}} {\bf{H_2P_2x_2}})+\mathcal{O}(snr))p_{x_1}({\bf{x_1}})p_{x_2}({\bf{x_2}}) \nonumber \\
=\exp{(-|{\bf{y}}|^2)} \times \nonumber \\
\mathbb{E}_{x_1,x_2}[{\bf{x_1}}] 
+\sqrt{snr}\mathbb{E}_{x_1,x_2}[{\bf{x_1}}]({\bf{y^{\dag}}} {\bf{H_1P_1x_1}}+({\bf{y^{\dag}}} {\bf{H_1P_1x_1}})^{\dag}) \nonumber \\
+\sqrt{snr}\mathbb{E}_{x_1,x_2}[{\bf{x_1}}]({\bf{y^{\dag}}} {\bf{H_2P_2x_2}}+({\bf{y^{\dag}}} {\bf{H_2P_2x_2}})^{\dag})+\mathcal{O}(snr)  \nonumber \\
=\exp{(-|{\bf{y}}|^2)} \times \nonumber \\
\mathbb{E}_{x_1,x_2}[{\bf{x_1}}] \nonumber \\
+\sqrt{snr}\mathbb{E}_{x_1,x_2}[{\bf{x_1}}]( ({{\bf{x_1}}^T{\bf{P_1}}^T{\bf{H_1}}^T{\bf{y}}^{*}})^T +({\bf{H_1P_1x_1}})^{\dag}{\bf{y}} ) \nonumber \\
+\sqrt{snr}\mathbb{E}_{x_1,x_2}[{\bf{x_1}}]( ({{\bf{x_2}}^T{\bf{P_2}}^T{\bf{H_2}}^T{\bf{y}}^{*}})^T +({\bf{H_2P_2x_2}})^{\dag}{\bf{y}} ) \nonumber \\
+\mathcal{O}(snr)  \nonumber \\
=\exp{(-|{\bf{y}}|^2)} \times \nonumber \\
\sqrt{snr}(\mathbb{E}_{x_1,x_2}[{\bf{x_1x_1^\dag}}]({\bf{H_1P_1}})^{\dag} {\bf{y}})+\mathcal{O}(snr) \nonumber \\
=\exp{(-|{\bf{y}}|^2)}\sqrt{snr}({\bf{H_1P_1}})^{\dag} {\bf{y}}+\mathcal{O}(snr)
\end{multline}
\normalsize
and the denominator of \eqref{nonlinearestimate1} is given by, 
\vspace{-0.1cm}
\scriptsize
\begin{multline}
\exp{(-|{\bf{y}}|^2)} \times \nonumber \\
(\sum_{x_1,x_2}(1+2\sqrt{snr}\mathcal{R}({\bf{y^{\dag}}} {\bf{H_1P_1x_1}}) \nonumber \\
+2\sqrt{snr}\mathcal{R}({\bf{y^{\dag}}} {\bf{H_2P_2x_2}}) \nonumber \\
+\mathcal{O}(snr))p_{x_1}({\bf{x_1}})p_{x_2}({\bf{x_2}})) \nonumber \\
=\exp{(-|{\bf{y}}|^2)} \times \nonumber \\
(1+\sqrt{snr}\sum_{x_1,x_2}({\bf{y}}^{\dag} {\bf{H_1P_1}x_1'} \nonumber \\
+({\bf{y}^{\dag}}{\bf{H_1P_1}x_1'})^{\dag})p_{x_1}({\bf{x_1}})p_{x_2}({\bf{x_2}})  \nonumber \\
+\sqrt{snr}\sum_{x_1,x_2}({\bf{y}^{\dag}} {\bf{H_2P_2}x_2'}+({\bf{y}}^{\dag} {\bf{H_2P_2}x_2'})^{\dag})+\mathcal{O}(snr)) \nonumber\\
=\exp{(-|{\bf{y}}|^2)}(1+ \nonumber \\
\sqrt{snr}({\bf{y}^{\dag}} \bf{H_1P_1}\mathbb{E}_{x_1}[{\bf{x_1}}]+\mathbb{E}_{x_1}[{\bf{x_1}^{\dag}}]({\bf{y}^{\dag}} {\bf{H_1P_1}})^{\dag}) \nonumber \\
\sqrt{snr}({\bf{y}^{\dag}} \bf{H_2P_2}\mathbb{E}_{x_2}[{\bf{x_2}}]+\mathbb{E}_{x_2}[{\bf{x_2}^{\dag}}]({\bf{y}^{\dag}} {\bf{H_2P_2}})^{\dag})+\mathcal{O}(snr)) \nonumber\\
=\exp{(-|{\bf{y}}|^2)}(1+\mathcal{O}(snr))
\end{multline}
\normalsize

Therefore,

\vspace{-0.1cm}
\footnotesize
\begin{multline}
\mathbb{E}_y[\mathbb{E}_{x_1|y}[{\bf{x_1|y}}]\left(\mathbb{E}_{x_1|y}[{\bf{x_1|y}}]\right)^{\dag}]=\\
\int_{{\bf{y}}\in\mathbb{C}^{n_r}}\frac{1}{\pi^{n_r}}\exp{\left(-|{\bf{y}}|^2\right)}\times \\ 
\quad \frac{\left(\sqrt{snr}\left({\bf{H_1P_1}}\right)^{\dag} {\bf{y}}+\mathcal{O}\left(snr\right)\right)\left(\sqrt{snr}\left({\bf{H_1P_1}}\right)^{\dag} {{\bf{y}}}+\mathcal{O}\left(snr\right)\right)^{\dag}}{1+\mathcal{O}\left(snr\right)}d{\bf{y}}\\ 
\end{multline}
\vspace{-0.35cm}
\begin{multline}
=\int_{{\bf{y}}\in\mathbb{C}^{n_r}}\frac{1}{\pi^{n_r}}\exp{\left(-|{\bf{y}}|^2\right)}\\
\quad\times\frac{snr\left({\bf{H_1P_1}}\right)^{\dag} {\bf{yy}^{\dag}} {\bf{H_1P_1}}+\mathcal{O}\left(snr^2\right)}{1+\mathcal{O}\left(snr\right)}d{\bf{y}}\\
\end{multline}
\normalsize
It follows that:
\vspace{-0.1cm}
\small
\begin{multline}
\mathbb{E}_y[\mathbb{E}_{x_1|y}[{\bf{x_1|y}}]\left(\mathbb{E}_{x_1|y}[{\bf{x_1|y}}]\right)^{\dag}] \\
=\left({\bf{H_1P_1}}\right)^{\dag} {\bf{H_1P_1}}snr+\mathcal{O}\left({snr}^2\right)
\end{multline}
\normalsize
Consequently, the low-snr expansion of the MMSE matrix of user 1 $\bf{E_1}$ is given as follows:
\begin{equation}
{\bf{E_1}}={\bf{I}}-({\bf{HP}})^{\dag}{\bf{H_1P_1}}.snr+\mathcal{O}(snr^{2})
\end{equation}
Similarly, the low-snr expansion of the MMSE matrix of user 2 $\bf{E_2}$ is given as follows:
\begin{equation}
{\bf{E_2}}={\bf{I}}-({\bf{H_2P_2}})^{\dag}{\bf{H_2P_2}}.snr+\mathcal{O}(snr^{2})
\end{equation}
Therefore, we can express the low-snr expansion of the total MMSE in terms of the snr as follows:
\small
\begin{multline}
MMSE\left(snr\right)=Tr\left\{{\bf{H_1P_1E_1}}\left({\bf{H_1P_1}}\right)^{\dag}\right\} \nonumber \\
+Tr\left\{{\bf{H_2P_2E_2}}\left({\bf{H_2P_2}}\right)^{\dag}\right\} \nonumber\\
=Tr\left\{{\bf{H_1P_1}}\left({\bf{I}}-\left({\bf{H_1P_1}}\right)^{\dag} {\bf{H_1P_1}}snr+\mathcal{O}\left(snr^2\right)\right)\left({\bf{H_1P_1}}\right)^{\dag}\right\} \nonumber \\
+Tr\left\{{\bf{H_2P_2}}\left({\bf{I}}-\left({\bf{H_2P_2}}\right)^{\dag} {\bf{H_2P_2}}snr+\mathcal{O}\left(snr^2\right)\right)\left({\bf{H_2P_2}}\right)^{\dag}\right\} \nonumber\\
=Tr\left\{{\bf{H_1P_1}}\left({\bf{H_1P_1}}\right)^{\dag}\right\}-Tr\left\{\left({\bf{H_1P_1}}\left({\bf{H_1P_1}}\right)^{\dag}\right)^2\right \} snr \nonumber \\
+Tr\left\{{\bf{H_2P_2}}\left({\bf{H_2P_2}}\right)^{\dag}\right\}-Tr\left\{\left({\bf{H_2P_2}}\left({\bf{H_2P_2}}\right)^{\dag}\right)^2\right \} snr \nonumber \\
+\mathcal{O}\left(snr^2\right)
\end{multline}
\normalsize
\vspace{-0.3cm}
\subsection{Derivation of the Multiuser Mutual Information at the Low SNR}
We shall now capitalize on the unveiled generalization of the fundamental relation between the mutual information and the MMSE plus covariance. Therefore, using similar steps to derive the low-snr expansion of the covariance given by,
\vspace{-0.1cm}
\small
\begin{multline}
\psi(snr)=\\
-Tr\left\{\bf{H_1P_1}\mathbb{E}_y[\bf{\mathbb{E}_{x_1|y}[\bf{x_1|y}]\mathbb{E}_{x_2|y}[\bf{x_2|y}]^{\dag}}]\bf{(H_2P_2)^{\dag}}\right\} \nonumber \\
-Tr\left\{\bf{H_2P_2}\mathbb{E}_y[\bf{\mathbb{E}_{x_2|y}[\bf{x_2|y}]\mathbb{E}_{x_1|y}[\bf{x_1|y}]^{\dag}}]\bf{(H_1P_1)^{\dag}}\right\},
\end{multline}
\normalsize
\vspace{-0.1cm}
Substituting the low-snr expansion of $\mathbb{E}_{x_1|y}[{\bf{x_1|y}}]$ and $\mathbb{E}_{x_2|y}[{\bf{x_2|y}}]$ into the covariance, the low-snr expansion of covariance $\psi(snr)$ as $snr \to 0$ is given by:
\vspace{-0.1cm}
\begin{multline}
\psi(snr)=\\
-Tr\left\{{\bf{H_1P_1}}\left({\bf{H_1P_1}}\right)^{\dag}{\bf{H_2P_2}}\left({\bf{H_2P_2}}\right)^{\dag}\right\}snr \nonumber \\
-Tr\left\{{\bf{H_2P_2}}\left({\bf{H_2P_2}}\right)^{\dag}{\bf{H_1P_1}}\left({\bf{H_1P_1}}\right)^{\dag}\right\}snr
\end{multline}
Therefore, capitalizing on the fundamental relation which states that,
\begin{equation}
\label{nfr00}
\frac{dI(snr)}{dsnr}=mmse(snr)+\psi(snr)
\end{equation}
The mutliuser mutual information at the low snr regime is the integral of both sides of \eqref{nfr00}, and so its given by:
\small
\begin{multline}
I(snr)=Tr \left\{\bf{H_1P_1}{(\bf{H_1P_1})}^{\dag}\right\}snr \\
+Tr \left\{\bf{H_2P_2}{(\bf{H_2P_2})}^{\dag}\right\}snr \\
- Tr\left\{(\bf{H_1P_1}{(\bf{H_1P_1})}^{\dag})^{2}\right\}{snr^2} \\
- Tr\left\{(\bf{H_2P_2}{(\bf{H_2P_2})}^{\dag})^{2}\right\}{snr^2}+ \\
+Tr\left\{\bf{H_1P_1}{(\bf{H_1P_1})}^{\dag}\bf{H_2P_2}{(\bf{H_2P_2})}^{\dag}\right\}{snr^2} \\
-Tr\left\{\bf{H_2P_2}{(\bf{H_2P_2})}^{\dag}\bf{H_1P_1}{(\bf{H_1P_1})}^{\dag}\right\}{snr^2}+\mathcal{O}(snr^3)
\end{multline}
\normalsize
Therefore, Theorem~\ref{LOWSNR} has been proved.
\bibliographystyle{IEEEtran}
\bibliography{IEEEabrv,mybibfile}
\end{document}